\documentclass[11pt]{article}
\usepackage{amsmath, amssymb, amsthm}
\usepackage{graphicx}
\usepackage{amsmath}
\usepackage{epsfig}
\usepackage{url}
\usepackage{algorithmic}
\usepackage{fullpage}

\newtheorem{theorem}{{Theorem}}
\newtheorem{property}{{Property}}
\newtheorem{lemma}{{Lemma}}

\title{A memory-efficient data structure representing exact-match overlap
graphs with application for next generation DNA assembly
\thanks{This work has been supported in
part by the following grants: NSF 0326155, NSF 0829916 and NIH
1R01GM079689-01A1.}}

\author{Hieu Dinh\\
University of Connecticut\\
\textsf{hdinh@engr.uconn.edu} \and
Sanguthevar Rajasekaran\\
University of Connecticut\\
\textsf{rajasek@engr.uconn.edu} }

\begin{document}

\maketitle              
\thispagestyle{empty}

\begin{abstract}

The maximal exact-match overlap of two strings $x$ and $y$, denoted
by $ov_{max}(x,y)$, is the longest string which is a suffix of $x$
and a prefix of $y$. The exact-match overlap graph of $n$ given
strings of length $\ell$ is an edge-weighted graph in which each
vertex is associated with a string and there is an edge $(x,y)$ of
weight $\omega = \ell - |ov_{max}(x,y)|$ if and only if $\omega \leq
\lambda$, where $|ov_{max}(x,y)|$ is the length of $ov_{max}(x,y)$
and $\lambda$ is a given threshold. In this paper, we show that the
exact-match overlap graphs can be represented by a compact data
structure that can be stored using at most $(2\lambda -1
)(2\lceil\log n\rceil + \lceil\log\lambda\rceil)n$ bits with a
guarantee that the basic operation of accessing an edge takes
$O(\log \lambda)$ time. We also propose two algorithms for
constructing the data structure for the exact-match overlap graph.
The first algorithm runs in $O(\lambda\ell n\log n)$ worse-case time
and requires $O(\lambda)$ extra memory. The second one runs in
$O(\lambda\ell n)$ time and requires $O(n)$ extra memory.

Exact-match overlap graphs have been broadly used in the context of
DNA assembly and the \emph{shortest super string problem} where the
number of strings $n$ ranges from a couple of thousands to a couple
of billions, the length $\ell$ of the strings is from 25 to 1000,
depending on DNA sequencing technologies. However, many DNA
assemblers using overlap graphs are facing a major problem of
constructing and storing them. Especially, it is impossible for
these DNA assemblers to handle the huge amount of data produced by
the next generation sequencing technologies where the number of
strings $n$ is usually very large ranging from hundred million to a
couple of billions. If a graph is explicitly stored, it would
require $\Omega(n^2)$ memory, which is impossible in practice in the
case that $n$ is greater than hundred million. In fact, to our best
knowledge there is no DNA assemblers that can handle such a large
number of strings. Fortunately, with our compact data structure, the
major problem of constructing and storing overlap graphs is
practically solved since it only requires linear time and and linear
memory. As a result, it opens the door of possibilities to build a
DNA assembler that can handle large-scale datasets efficiently.

\end{abstract}

\newpage
\pagenumbering{arabic}

\section{Introduction} \label{sec:intro}

An exact-match overlap graph of $n$ given strings of length $\ell$
is an edge-weighted graph defined informally as follows. Each vertex
is associated with a string and there is an edge $(x,y)$ of weight
$\omega = \ell - |ov_{max}(x,y)|$ if and only if $\omega \leq
\lambda$, where $\lambda$ is a given threshold and $|ov_{max}(x,y)|$
is the length of the maximal exact-match overlap of two strings $x$
and $y$. The formal definition of the exact-match overlap graph is
given in Section \ref{sec:prelim}.

Storing the exact-match overlap graphs efficiently in term of memory
becomes essential when the number of strings is very large. In the
literature, there are two common data structures to store a general
graph $G=(V,E)$. The first data structure uses a two-dimensional
array of size $|V|\times |V|$. We call it array-based data
structure. One of its advantages is that the time of accessing a
given edge is $O(1)$. However, it requires $\Omega(|V|^2)$ memory.
The second data structure stores the set of edges $E$. We call it
edge-based data structure. Of course, it requires $\Omega(|V|+|E|)$
memory and the time of accessing a given edge is
$O(\log\Delta)$, where $\Delta$ is the degree of the graph. Both of
these data structures require $\Omega(|E|)$ memory. If the
exact-match overlap graphs are stored by these two data structures,
we will need  $\Omega(|E|)$ memory. Even this much of memory may not
be feasible in the case that the number of strings is over hundred
millions. In this paper we focus on data structures for the
exact-match overlap graphs that will need much less memory than
$|E|$.

\subsection{Our contributions}
We show that there is a compact data structure representing the
exact-match overlap graph that needs much less memory than $|E|$ with a
guarantee that the basic operation of accessing an edge takes
$O(\log \lambda)$ time, which is almost a constant in the context of
DNA assembly. The data structure can be constructed efficiently in
time and memory as well. In particular, we show that
\begin{itemize}
  \item The data structure takes no more than $(2\lambda -1 )(2\lceil\log n\rceil +
  \lceil\log\lambda\rceil)n$ bits.
  \item The data structure can be constructed in $O(\lambda\ell n)$ time.
\end{itemize}

As a result, any algorithm using overlap graphs can be simulated by
our compact data structure with no more $(2\lambda -1 )(2\lceil\log
n\rceil +  \lceil\log\lambda\rceil)n$ bits for storing the overlap
graph and paying extra $O(\log\lambda)$ time factor overhead.
Apparently, if $\lambda$ is a constant or much much smaller than
$n$, our data structure will be a perfect solution for any
application that does not have enough memory for storing the overlap
graph in traditional way.

Our claim may sound contradictory because in some exact-match
overlap graphs the number of edges can be $\Omega(n^2)$ and it seems
like it should require at least $\Omega(n^2)$ time and memory to
construct them. Fortunately, because of some special properties of
the exact-match overlap graphs, we can construct and store them
efficiently. In Section \ref{sec:data-struct-ov-graph}, we will
describe these special properties in detail.

Briefly, the idea of storing the overlap graph compactly is from the
following simple observation. If the strings are sorted in the
lexicographic order, then for any string $x$ the lexicographic
orders of the strings that contain $x$ as a prefix are in a certain
integer range or integer interval $[a,b]$. Therefore, the
information about out-neighborhood of a vertex can be described by
at most $\lambda$ intervals. Such intervals have a nice property
that they are either disjoint or contain each other. This property
allows us to describe the out-neighborhood of a vertex by at most
$2\lambda - 1$ disjoint intervals. Each interval costs $2\lceil\log
n\rceil + \lceil\log\lambda\rceil$ bits, where $2\lceil\log n\rceil$
bits are for storing its two bounds and $\lceil\log\lambda\rceil$
bits are for storing the weight. We have $n$ vertices so the amount
of memory required by our data structure is no more than $(2\lambda
-1 )(2\lceil\log n\rceil + \lceil\log\lambda\rceil)n$ bits. Note
that this is just an upper bound. In practice, the amount of memory
may be much less than that.

\subsection{Application: DNA assembly}
The main motivation for the exact-match overlap graphs comes from
their use in implementing fast approximation algorithms for the
\emph{shortest super string problem} which is the very first problem
formulation for DNA assembly. The exact-match overlap graphs can be
used for other problem formulations for DNA assembly as well.

Exact-match overlap graphs have been broadly used in the context of
DNA assembly and the \emph{shortest super string problem} where the
number of strings $n$ ranges from a couple of thousands to a couple
of billions, the length $\ell$ of the strings is from 25 to 1000,
depending on DNA sequencing technologies. However, many DNA
assemblers using overlap graphs are facing a major problem of
constructing and storing them. Especially, it is impossible for
these DNA assemblers to handle the huge amount of data produced by
the next generation sequencing technologies where the number of
strings $n$ is usually very large ranging from hundred million to a
couple of billions. If a graph is explicitly stored, it would
require $\Omega(n^2)$ memory, which is impossible in practice in the
case that $n$ is greater than hundred million. In fact, to our best
knowledge there is no DNA assemblers that can handle such a large
number of strings. Fortunately, with our compact data structure, the
major problem of constructing and storing overlap graphs is
practically solved since it only requires linear time and linear
memory. As a result, it opens the door of possibilities to build a
DNA assembler that can handle large-scale datasets efficiently.

\subsection{Related work} \label{sec:related-work}
Gusfield et al. \cite{DanGusfield-All-Pairs1992},
\cite{DanGusfield-Book1997} consider the\emph{ all-pairs
suffix-prefix problem} which is actually a special case of computing
the exact-match overlap graphs when $\lambda = \ell$. They devised
an $O(\ell n + n^2)$ time algorithm for solving the all-pairs
suffix-prefix problem. In this case, the exact-match overlap graph
is a complete graph. So the run time of the algorithm is optimal if
the exact-match overlap graph is stored in the common way.

Although the run time of the algorithm by Gusfield et al. is
theoretically optimal in that setting, it uses the generalized
suffix tree which has two disadvantages in practice. The first
disadvantage is that the space consumption of the suffix tree is
quite large \cite{Kurtz-Suffix-Trees1999}. The second disadvantage
is that the suffix tree usually suffers from a poor locality of
memory references \cite{Ohlebusch-All-Pairs2010}. Fortunately,
Abouelhoda et al. \cite{Abouelhoda-Replace-Suffix-Trees2004}
proposed a suffix tree simulation framework that allows any
algorithm using the suffix tree to be simulated by enhanced suffix
arrays. Ohlebusch and Gog \cite{Ohlebusch-All-Pairs2010} made use of
properties of the enhanced suffix arrays to devise an algorithm for
solving the all-pairs suffix-prefix problem directly without using
the suffix tree simulation framework. The run time of the algorithm
by Ohlebusch and Gog is also $O(\ell n + n^2)$. Please note that our
data structure and algorithm can be used to solve the suffix-prefix
problem in $O(\lambda \ell n)$ time. In the context of DNA assembly,
$\lambda$ is typically much smaller than $n$ and hence our algorithm
will be faster than the algorithms of \cite{DanGusfield-Book1997}
and \cite{Ohlebusch-All-Pairs2010}.

In the literature, exact-match overlap graphs should be
distinguished from approximate-match overlap graphs which is
considered in \cite{GeneMyers-StringGraph2005},
\cite{Medvedev07computabilityof},
\cite{MihaiPop-AssemblySurvey2009}. In the approximate-match overlap
graph, there is an edge between two strings $x$ and $y$ if and only
if there is a prefix of $x$, say $x'$, and there is a suffix of $y$,
say $y'$, such that the edit distance between $x'$ and $y'$ is no
more than a certain threshold.

\section{Preliminaries} \label{sec:prelim}
Let $\Sigma$ be the alphabet. The size of $\Sigma$ is a constant. In
the context of DNA assembly, $\Sigma = \{A, C, G, T\}$. The length
of a string $x$ on $\Sigma$, denoted by $|x|$, is the number of
symbols in $x$. Let $x[i]$ be the $i$-th symbol of string $x$, and
$x[i,j]$ be the substring of $x$ between the $i$-th and the $j$
positions. A prefix of string $x$ is the substring $x[1,i]$ for some
$i$. A suffix of string $x$ is the substring $x[i,|x|]$ for some
$i$.

Given two strings $x$ and $y$ on $\Sigma$, an \emph{exact-match
overlap} between $x$ and $y$, denoted by $ov(x,y)$, is a string
which is a suffix of $x$ and a prefix of $y$ (notice that this
definition is not symmetric). The maximal exact-match overlap
between $x$ and $y$, denoted by $ov_{max}(x,y)$, is the longest
exact-match overlap between $x$ and $y$.

\textbf{Exact-match overlap graphs:} Given $n$ strings
$s_1,s_2,\dots,s_n$ and a threshold $\lambda$, the exact-match
overlap graph is an edge-weighted directed graph $G=(V,E)$ in which
there is a vertex $v_i \in V$ associated with the string $s_i$, for $1\leq i\leq n$. There
is an edge $(v_i,v_j) \in E$ if and only if $|s_i| -
|ov_{max}(s_i,s_j)| \leq \lambda$. The weight of the edge
$(v_i,v_j)$, denoted by $\omega(v_i,v_j)$, is $|s_i| -
|ov_{max}(s_i,s_j)|$.

\begin{figure}[h]
\begin{center}
  \includegraphics[width=3in]{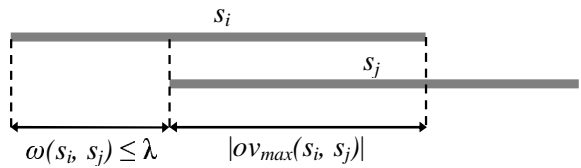}\\
  \caption{An example of an overlap edge.}\label{fig:overlap-edge}
\end{center}
\end{figure}

The set of out-neighbors of a vertex $v$ is denoted by
$OutNeigh(v)$. The size of the set of out-neighbors of $v$,
$|OutNeigh(v)|$, is called the out-degree of $v$. We denote the out-degree of
$v$ as $deg_{out}(v) = |OutNeigh(v)|$.

For simplicity, we assume that all the strings $s_1,s_2,\dots,s_n$
have the same length $\ell$. Otherwise, let $\ell$ be the length of
the longest string and all else works.

\textbf{The operation of accessing an edge given its two endpoints:}
Given any two vertices $v_i$ and $v_j$, the operation of accessing
the edge $(v_i,v_j)$ is the task of returning $\omega(v_i,v_j)$ if
$(v_i,v_j)$ is actually an edge of the graph, and returning $NULL$
if $(v_i,v_j)$ is not.

\section{A memory-efficient data structure representing an exact-match overlap graph}
\label{sec:data-struct-ov-graph}  In this section, we describe a
memory-efficient data structure to store an exact-match overlap
graph. It only requires at most $(2\lambda -1 )(2\lceil\log n\rceil
+ \lceil\log\lambda\rceil)n$ bits. It guarantees that the
time for accessing an edge, given two endpoints of the edge, is
$O(\log \lambda)$. This may sound like a contradictory claim because in some
exact-match overlap graphs the number of edges can be $\Omega(n^2)$ and it seems like
it should require at least $\Omega(n^2)$ time and space to construct
them. Fortunately, because of some special properties of the
exact-match overlap graphs, we can construct and store them
efficiently. In the following paragraphs, we will describe these
special properties.

Without loss of generality, we assume that the $n$ input strings
$s_1,s_2,\dots,s_n$ are sorted in lexicographic order. We can assume
this because if they are not sorted, we can sort them by using the
radix sort algorithm which runs in $O(\ell n)$ time. The algorithm
radix sort takes $O(\ell n)$ time in this case because we consider
the constant alphabet size. Otherwise, it would take additional
$O(|\Sigma|\log(|\Sigma|))$ time to sort the alphabet.

Each string $s_i$ and its corresponding vertex $v_i$ in the
exact-match overlap graph are determined by the string's
lexicographic order $i$. We refer to the lexicographic order of any
string as its \emph{identification number}. We will access an input
string and its vertex through its identification number. Therefore,
the identification number and the vertex of an input string are used
interchangeably. Also, it is not hard to see that we need $\lceil
\log n \rceil$ bits to store an identification number. We have the
following properties.

Given an arbitrary string $x$, let $PREFIX(x)$ be the set of
identification numbers such that $x$ is a prefix of their
corresponding input strings. Formally, $PREFIX(x) = \{i | x$ is a
prefix of $s_i\}$.

\begin{property} If $PREFIX(x)
\neq \emptyset$, then $PREFIX(x) = [a,b]$, where $[a,b]$ is some
integer interval containing integers $a, a+1,\dots,b-1,b$.
\label{property:interval}
\end{property}

\begin{proof} Let $a = \min_{i \in PREFIX(x)}i$ and $b = \max_{i \in
PREFIX(x)}i$. Clearly, $PREFIX(x) \subseteq [a,b]$. On the other
hand, we will show that $[a,b] \subseteq PREFIX(x)$. Let $i$ be any
identification number in the interval $[a,b]$. Since the input
strings are in lexicographically sorted order, $s_a[1,|x|] \leq
s_i[1,|x|] \leq s_b[1,|x|]$. Since $a \in PREFIX(x)$ and $b \in
PREFIX(x)$, $s_a[1,|x|] = s_b[1,|x|]$. Thus, $s_a[1,|x|] =
s_i[1,|x|] = s_b[1,|x|]$. Therefore, $x$ is a prefix of $s_i$.
Hence, $i \in PREFIX(x)$.
\end{proof}

For example, let
\begin{eqnarray}
s_1 & = & AAACCGGGGTTT \nonumber \\
s_2 & = & ACCCGAATTTGT \nonumber \\
s_3 & = & ACCCTGTGGTAT \nonumber \\
s_4 & = & ACCGGCTTTCCA \nonumber \\
s_5 & = & ACTAAGGAATTT \nonumber \\
s_6 & = & TGGCCGAAGAAG \nonumber
\end{eqnarray}

If $x=AC$, then $PREFIX(x) = [2,5]$. Similarly, if $x=ACCC$, then
$PREFIX(x) = [2,3]$.

Property \ref{property:interval} tells us that $PREFIX(x)$ can be
expressed by an interval which is determined by its lower bound and
its upper bound. So we only need $2\lceil \log n \rceil$ bits to
store $PREFIX(x)$. In the rest of this paper, we will
refer to $PREFIX(x)$ as an interval. Also, given an identification
number $i$, checking whether $i$ is in $PREFIX(x)$ can be done in
$O(1)$ time. In the subsection \ref{sec:comp-PREFIX}, we will
discuss two algorithms computing $PREFIX(x)$, for a given string
$x$. The run times of these algorithms are $O(|x| \log n)$ and $O(|x|)$, respectively.

Property \ref{property:interval} leads to the following property.

\begin{property} $OutNeigh(v_i) =
\bigcup_{1 \leq \omega \leq \lambda} PREFIX(s_i[\omega + 1,|s_i|])$
for each vertex $v_i$. In the other words, $OutNeigh(v_i)$ is the
union of at most $\lambda$ non-empty intervals.
\label{property:out-neigh-overlap}
\end{property}

\begin{proof} Let $v_j$ be a vertex in $OutNeigh(v_i)$. By
the definition of the exact-match overlap graph, $1 \leq |s_i| -
|ov_{max}(s_i,s_j)| = \omega(v_i,v_j) \leq \lambda$. Let
$\omega(s_i,s_j) = \omega.$ Therefore, $ov_{max}(s_i,s_j) =
s_i[\omega + 1,|s_i|] = s_j[1,|ov_{max}(s_i,s_j)|]$. This implies
$v_j \in PREFIX(s_i[\omega + 1,|s_i|])$.

On the other hand, let $v_j$ be any vertex in $PREFIX(s_i[\omega +
1,|s_i|])$, it is easy to check that $v_j \in OutNeigh(v_i)$. Hence,
$OutNeigh(v_i) = \bigcup_{1 \leq \omega \leq \lambda}
PREFIX(s_i[\omega + 1,|s_i|])$.
\end{proof}

From Property \ref{property:out-neigh-overlap}, it follows that we
can represent $OutNeigh(v_i)$ by at most $\lambda$ non-empty
intervals, which need at most $2\lambda \lceil \log n \rceil$ bits
to store. Therefore, it takes at most $2n\lambda \lceil \log n
\rceil$ bits to store the exact-match overlap graph. However, given
two vertices $v_i$ and $v_j$, it takes $O(\lambda)$ time to retrieve
$\omega(v_i,v_j)$ because we have to sequentially check if $v_j$ is
in $PREFIX(s_i[2,|s_i|])$,$PREFIX(s_i[3,|s_i|])$,$\dots$,\\
$PREFIX(s_i[\lambda+1,|s_i|])$. But if $OutNeigh(v_i)$ can be
represented by $k$ disjoint intervals then the task of retrieving
$\omega(v_i,v_j)$ can be done in $O(\log k)$ time by using binary
search. In Lemma \ref{theorem:number-of-intvl-01}, we show that
$OutNeigh(v_i)$ is a union of at most $2\lambda - 1$ disjoint
intervals.

\begin{property} For any two strings $x$ and $y$ with $|x| < |y|$,
then either one of the two following statements is true:
\begin{itemize}
  \item $PREFIX(y) \subseteq PREFIX(x)$
  \item $PREFIX(y) \bigcap PREFIX(x) = \emptyset$
\end{itemize}
\label{property:disjoint-or-containing-intvl}
\end{property}

\begin{proof} There are only two possible cases that can happen to $x$ and $y$.
\\\textbf{Case 1}: $x$ is a prefix of $y$. For this case, it is not
hard to infer that $PREFIX(y) \subseteq PREFIX(x)$.
\\\textbf{Case 2}: $x$ is not a prefix of $y$. For this case, it is not
hard to infer that $PREFIX(y) \bigcap PREFIX(x) = \emptyset$.
\end{proof}

\begin{lemma}Given ${\lambda}$ intervals $[a_1,b_1],[a_2,b_2]\dots [a_{\lambda},b_{\lambda}]$
satisfied Property \ref{property:disjoint-or-containing-intvl}, the
union of them is the union of at most $2{\lambda} - 1$ disjoint
intervals. Formally, there exist $p \leq 2{\lambda} - 1$ disjoint
intervals $[a'_1,b'_1],[a'_2,b'_2]\dots [a'_p,b'_p]$ such that
$\bigcup_{1 \leq i \leq {\lambda}} [a_{i}, b_{i}] = \bigcup_{1 \leq
i \leq p} [a'_{i}, b'_{i}]$. \label{lemma:number-of-intvl}
\end{lemma}

\begin{proof}We say interval $[a_i,b_i]$ is a parent of
interval $[a_j,b_j]$ if $[a_i,b_i]$ is the smallest interval
containing $[a_j,b_j]$. We also say interval $[a_j,b_j]$ is a child
of interval $[a_i,b_i]$. Since the intervals $[a_i,b_i]$ are either
pairwise disjoint or contain each other, each interval has at most
one parent. Therefore, the set of the intervals $[a_i,b_i]$ form a
forest in which each vertex is associated with an interval, see
Figure \ref{fig:interval-tree}. For each interval $[a_i,b_i]$, let
$I_i$ be the set of the maximal intervals that are contained in
interval $[a_i,b_i]$ but disjoint with all of its children. For
example, if $[a_i,b_i] = [1,20]$ and its child intervals are $[3,5],
[7,8]$ and $[12,15]$, then $I_i = \{[1,2],[6,6],[9,11],[16,20]\}$.
In the case the interval $[a_i,b_i]$ is a leaf interval (an interval
does not have any children), $I_i$ is simply the set containing only
interval $[a_i,b_i]$. Let $A = \bigcup_{1 \leq i \leq {\lambda}}
I_i$. We will show that $A$ is the set of the disjoint intervals
$[a'_i,b'_i]$ satisfying the condition of the lemma.

Firstly, we show that $\bigcup_{1 \leq i \leq {\lambda}}[a_i,b_i] =
\bigcup_{[a'_i,b'_i] \in A}[a'_i,b'_i]$. By the construction of
$I_i$, it is trivial to see that $\bigcup_{[a'_i,b'_i] \in
A}[a'_i,b'_i] \subseteq \bigcup_{1 \leq i \leq {\lambda}}[a_i,b_i]$.
Conversely, it is enough to show that $[a_i,b_i] \subseteq
\bigcup_{[a'_i,b'_i] \in A}[a'_i,b'_i]$ for any $1 \leq i \leq
{\lambda}$. This can be proved by induction on vertices in each tree
of the forest. For the base case, obviously each leaf interval
$[a_i,b_i]$ is in $A$. Therefore, $[a_i,b_i] \subseteq
\bigcup_{[a'_i,b'_i] \in A}[a'_i,b'_i]$ for any leaf interval
$[a_i,b_i]$. For any internal interval $[a_i,b_i]$, assume that all
of its child intervals are subsets of $\bigcup_{[a'_i,b'_i] \in
A}[a'_i,b'_i]$. By the construction of $I_i$, $[a_i,b_i]$ is a union
of all of the intervals in $I_i$ and all of its child intervals.
Therefore, $[a_i,b_i] \subseteq \bigcup_{[a'_i,b'_i] \in
A}[a'_i,b'_i]$.

Secondly, we show that the intervals in $A$ are pairwise disjoint.
It is sufficient to show that any interval in $I_i$ is disjoint with
every interval in $I_j$ for $i \neq j$. Obviously, the statement is
true if $[a_i,b_i] \cap [a_j,b_j] = \emptyset$. Let us consider the
case where one contains the other. Without loss of generality, we assume
that $[a_j,b_j] \subset [a_i,b_i]$. Consider two cases: \\
\textbf{Case 1:} $[a_i,b_i]$ is the parent of $[a_j,b_j]$. By the
construction of $I_i$, any interval in $I_i$ is disjoint with
$[a_j,b_j]$. By the construction of $I_j$, any interval in $I_j$ is
contained in $[a_j,b_j]$. Therefore, they are disjoint. \\
\textbf{Case 2:} $[a_i,b_i]$ is not the parent of $[a_j,b_j]$. Let
$[a_j,b_j] = [a_{i_0},b_{i_0}] \subset [a_{i_i},b_{i_i}] \dots
\subset [a_{i_h},b_{i_h}] = [a_i,b_i]$, where
$[a_{i_{t}},b_{i_{t}}]$ is the parent of
$[a_{i_{t-1}},b_{i_{t-1}}]$. From the result in the Case 1, any
interval in $I_{i_t}$ is disjoint with $[a_{i_{t-1}},b_{i_{t-1}}]$
for $1 \leq t \leq h$. So any interval in $I_i$ is disjoint with
$[a_j,b_j]$. We already know that any interval in $I_j$ is contained
in $[a_j,b_j]$. Thus, they are disjoint.

Finally, we show that the number of intervals in $A$ is no more than
$2{\lambda} - 1$. We have $|A| = \sum_{i=1}^{{\lambda}} |I_i|$. It
is easy to see that the number of intervals in $I_i$ is no more than
the number of children of $[a_i,b_i]$ plus one, which is equal to
the degree of the vertex associated with $[a_i,b_i]$ if the vertex
is not a root of a tree in the forest, and equal to the degree of
the vertex plus one if the vertex is a root. Let $q$ be the number
of trees in the forest. Then, $|A|=\sum_{i=1}^{{\lambda}} |I_i| \leq
\sum_{i=1}^{{\lambda}}d_i + q = 2|E| + p$, where $d_i$ is the degree
of the vertex associated with $[a_i,b_i]$ and $E$ is the set of the
edges of the forest. We know that in a tree the number of edges is
equal to the number of vertices minus one. Thus, $|E| = {\lambda} -
q$. Therefore, $|A| \leq 2{\lambda} - q \leq 2{\lambda} - 1$. This
completes our proof.
\end{proof}

\begin{figure}[h]
\begin{center}
  \includegraphics[width=3.2in]{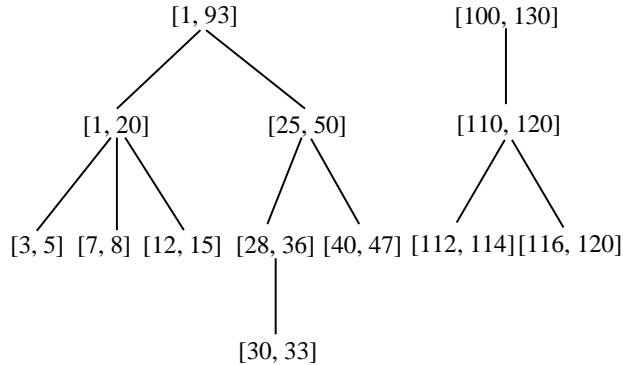}\\
  \caption{A forest illustration in the proof of Lemma \ref{lemma:number-of-intvl}.}\label{fig:interval-tree}
\end{center}
\end{figure}

From the proof, an algorithm computing the disjoint intervals is
straightforward by first constructing the interval forest. Once the
forest is built, outputting the disjoint intervals can be done
easily at each vertex. However, designing a fast algorithm for
constructing the forest is not trivial. In the subsection
\ref{sec:disjoint-intvl}, we will discuss an $O({\lambda}\log
{\lambda})$-time algorithm for constructing the forest. Thereby,
there is an $O({\lambda}\log {\lambda})$-time algorithm for
computing the disjoint intervals $[a'_i,b'_i]$ in Lemma
\ref{lemma:number-of-intvl}, given ${\lambda}$ intervals satisfying
Property \ref{property:disjoint-or-containing-intvl}. Also, from
Property \ref{property:disjoint-or-containing-intvl} and Lemma
\ref{theorem:number-of-intvl-01}, it is not hard to prove the
following theorem.

\begin{theorem}$OutNeigh(v_i)$ is the union of at most $2\lambda - 1$
disjoint intervals. Formally, $OutNeigh(v_i) = \bigcup_{1 \leq m
\leq p} [a_{m}, b_{m}]$ where $p \leq 2\lambda - 1$, $[a_{m}, b_{m}]
\bigcap [a_{m'}, b_{m'}] = \emptyset$ for $1 \leq m \neq m' \leq p$.
Furthermore, $\omega(v_i,v_j) = \omega(v_i,v_k)$ for any $1 \leq m
\leq p$ and for any $v_i,v_k \in
[a_m,b_m]$.\label{theorem:number-of-intvl-01}
\end{theorem}

Theorem \ref{theorem:number-of-intvl-01} suggests a way of
storing $OutNeigh(v_i)$ by at most $(2\lambda - 1)$ disjoint
intervals. Each interval takes $2\lceil \log n \rceil$ bits to
store its lower bound and its upper bound, and $\lceil \log \lambda
\rceil$ bits to store the weight. Thus, we need $2\lceil \log n
\rceil + \lceil \log \lambda \rceil$ to store each interval.
Therefore, it takes at most $(2\lambda - 1)(2\lceil \log n \rceil +
\lceil \log \lambda \rceil)$ bits to store each $OutNeigh(v_i)$.
Overall, we need $(2\lambda - 1)(2\lceil \log n \rceil + \lceil \log
\lambda \rceil)n$ bits to store the exact-match overlap graph. Of
course, the disjoint intervals of each $OutNeigh(v_i)$ are stored in
the sorted order of their lower bounds. Therefore, the operation of
accessing an edge $(v_i,v_j)$ can be easily done in $O(\log
\lambda)$ time by using binary search.

\section{Algorithms for constructing the compact data
structure} \label{sec:construct-data-struct} In this section, we
describe two algorithms for constructing the data structure representing
the exact-match overlap graph. The run time of the first
algorithm is $O(\lambda\ell n\log n)$ and it only uses $O(\lambda)$
extra memory, besides $\ell n \lceil\log|\Sigma|\rceil$ bits memory
used to store the $n$ input strings. The second algorithm runs in
$O(\lambda\ell n)$ time and requires $O(n)$ extra memory. As shown
in Section \ref{sec:data-struct-ov-graph}, the algorithms need two
routines. The first routine computes $PREFIX(x)$ and the second one
computes the disjoint intervals described in Lemma
\ref{lemma:number-of-intvl}.

\subsection{Computing interval $PREFIX(x)$} \label{sec:comp-PREFIX}
In this subsection, we consider the problem of computing the interval
$PREFIX(x)$, given a string $x$ and $n$ input strings
$s_1,s_2,\dots, s_n$ of the same length $\ell$ in lexicographical order. We
describe two algorithms for this problem. The first algorithm takes
$O(|x|\log n)$ time and $O(1)$ extra memory. The second algorithm
runs in $O(|x|)$ time and requires $O(n)$ extra memory.

\subsubsection{A binary search based algorithm} \label{sec:algo-PREFIX-BSearch}
Let $[a_i,b_i] = PREFIX(x[1,i])$ for $1 \leq i \leq |x|$. It is easy
to see that $PREFIX(x) = [a_{|x|},b_{|x|}] \subseteq [a_{|x| -
1},b_{|x| - 1}] \subseteq \dots \subseteq [a_1,b_1]$. Consider the
following input strings, for example.
\begin{eqnarray}
s_1 & = & AAACCGGGGTTT \nonumber \\
s_2 & = & ACCAGAATTTGT \nonumber \\
s_3 & = & ACCATGTGGTAT \nonumber \\
s_4 & = & ACGGGCTTTCCA \nonumber \\
s_5 & = & ACTAAGGAATTT \nonumber \\
s_6 & = & TGGCCGAAGAAG \nonumber \\
x & = & ACCA \nonumber
\end{eqnarray}
Then, $[a_1,b_1] = [1,5], [a_2,b_2] = [2,5], [a_3,b_3] = [2,3]$ and
$PREFIX(x) = [a_4,b_4] = [2,3]$.

We will find $[a_{i},b_{i}]$ from $[a_{i-1},b_{i-1}]$ for $i$ from
$1$ to $|x|$, where $[a_0,b_0]$ = [1,n] initially. Thereby,
$PREFIX(x)$ is computed. Let $Col_i$ be the string that consists of
all the symbols at position $i$ of the input strings. In the above
example, $Col_3 = ACCGTG$. Observe that the symbols in string
$Col_i[a_{i-1},b_{i-1}]$ are in lexicographical order for $1 \leq i \leq
|x|$. Thus, any symbol in the string $Col_i[a_{i-1},b_{i-1}]$ appears
consecutively. Another observation is that $[a_i,b_i]$ is the
interval where the symbol $x[i]$ appears consecutively in string
$Col_i[a_{i-1},b_{i-1}]$. Therefore, $[a_{i},b_{i}]$ is determined
by searching for the symbol $x[i]$ in the string $Col_i[a_{i-1},b_{i-1}]$. This
can be done easily by first using the binary search to find a
position in the string $Col_i[a_{i-1},b_{i-1}]$ where the symbol $x[i]$
appears. If the symbol $x[i]$ is not found, we return the empty interval and
stop. If the symbol $x[i]$ is found at position $c_i$, then $a_i$
(respectively $b_i$) can be determined by using the double search routine
in string $Col_i[a_{i-1},c_i]$ (resp. string $Col_i[c_i,b_{i-1}]$)
as follows. We consider the symbols in the string $Col_i[a_{i-1},c_i]$ at
positions $c_i - 2^0, c_i - 2^1,\dots,c_i - 2^k, a_{i-1}$, where
$k=\lfloor\log(c_i - a_{i-1})\rfloor$. We find $j$ such that the symbol
$Col_i[c_i - 2^j]$ is the symbol $x[i]$ but the symbol $Col_i[c_i -
2^{j+1}]$ is not. Finally, $a_i$ is determined by using binary
search in string $Col_i[c_i - 2^j, c_i - 2^{j+1}]$.  Similarly,
$b_i$ is determined. The pseudo-code is given as follows.

\begin{algorithmic}[1]
\STATE {Initialize $[a_0,b_0] = [1,n]$.}
\FOR {$i=1$ to $|x|$}
    \STATE {Find the symbol $x[i]$ in the string $Col_i[a_{i-1},b_{i-1}]$ using binary search.}
    \IF {the symbol $x[i]$ appears in the string $Col_i[a_{i-1},b_{i-1}]$}
        \STATE {Let $c_i$ be the position of the symbol $x[i]$ returned by the binary search.}
        \STATE {Find $a_i$ by double search and then binary search in the string $Col_i[a_{i-1},c_i]$.}
        \STATE {Find $b_i$ by double search and then binary search in the string $Col_i[c_i,b_{i-1}]$.}
    \ELSE
        \STATE{Return the empty interval $\emptyset$.}
    \ENDIF
\ENDFOR
\STATE{Return the interval $[a_{|x|},b_{|x|}]$.}
\end{algorithmic}

\emph{Analysis:} As we discussed above, it is easy to see the
correctness of the algorithm. Let us analyze the memory and time
complexity of the algorithm. Since the algorithm only uses binary
search and double search, it needs $O(1)$ extra memory. For time
complexity, it is easy to see that computing the interval $[a_i,b_i]$ at
step $i$ takes $O(\log(b_{i-1} - a_{i-1})) \leq O(\log n)$ time
because both binary search and double search take $O(\log(b_{i-1} -
a_{i-1}))$ time. Overall, the algorithm takes $O(|x|\log n)$
time because there are at most $|x|$ steps.

\subsubsection{A trie-based algorithm} \label{sec:algo-PREFIX-Trie}

As we have seen in Subsection \ref{sec:algo-PREFIX-BSearch}, to
compute the interval $[a_i,b_i]$ for symbol $x[i]$, we use binary search
to find the symbol $x[i]$ in the interval $[a_{i-1},b_{i-1}]$. The
binary search takes $O(\log (b_{i-1} -a_{i-1})) \leq O(\log
n)$ time. We can reduce the $O(\log n)$ factor to $O(1)$ in
computing the interval $[a_i,b_i]$ by pre-computing all of the
intervals for each symbol in the alphabet $\Sigma$ and store them in
a trie. Given the symbol $x[i]$, to find the interval $[a_i,b_i]$ we just
retrieve it from the trie, which takes $O(1)$ time. The trie is
defined as follows (see Figure \ref{fig:trie}). At each node in the
trie, we store a symbol and its interval. Observe that we do not
have to store the nodes that have only one child. These nodes form
chains in the trie. We will remove such chains and store their
lengths in each remaining node. As a result, each internal node in
the trie has at least two children. Because each internal node has
at least two children, the number of nodes in the trie is no more
than twice the number of leaves, which is equal to $2n$. Therefore,
we need $O(n)$ memory to store the trie. Also, it is well-known that
the trie can be constructed recursively in $O(\ell n)$ time.

\begin{figure}[h]
\begin{center}
  \includegraphics[width=3in]{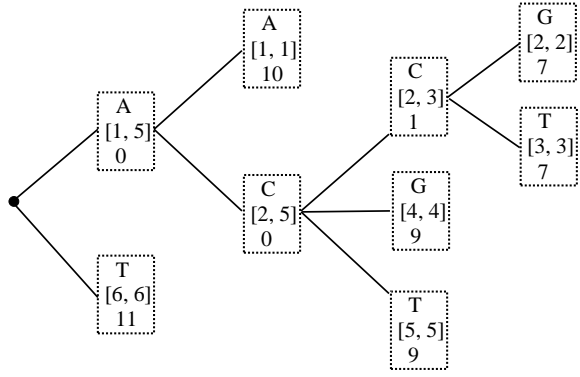}\\
  \caption{An illustration of a trie for the example input strings in Subsection \ref{sec:algo-PREFIX-BSearch}.}\label{fig:trie}
\end{center}
\end{figure}

It is easy to see that once the trie is constructed, the task of
finding the interval $[a_i,b_i]$ for each symbol $x[i]$ takes $O(1)$
time. Therefore, computing $PREFIX(x)$ will take $O(|x|)$
time.

\subsection{Computing the disjoint intervals} \label{sec:disjoint-intvl}
In this subsection, we consider the problem of computing the maximal
disjoint intervals, given $k$ intervals
$[a_1,b_1],[a_2,b_2],$
$\dots,[a_k,b_k]$ which either are pairwise
disjoint or contain each other. As discussed in Section
\ref{sec:data-struct-ov-graph}, it is sufficient to build the forest
of the $k$ input intervals. Once the forest is built, outputting the
maximal disjoint intervals can be done easily at each vertex of the
forest.

The algorithm for the problem is described as follows. First we sort
the input intervals in non-decreasing order of their lower bounds
$a_i$. Among those intervals whose lower bounds are equal, we sort
them in decreasing order of their upper bounds $b_i$. So after this
step, we have 1) $a_1 \leq a_2 \leq \dots \leq a_k$ and 2) if $a_i =
a_j$ then $b_i > b_j$ for $1 \leq i < j \leq k$. Since the input
intervals either are pairwise disjoint or contain each other, there
are only two possibilities happening to two intervals $[a_i,b_i]$
and $[a_{i+1},b_{i+1}]$ for $1 \leq i < k $. Either $[a_i,b_i]$
contains $[a_{i+1},b_{i+1}]$ or they are disjoint. Observe that if
$[a_i,b_i]$ contains $[a_{i+1},b_{i+1}]$, then $[a_i,b_i]$ is
actually the parent of $[a_{i+1},b_{i+1}]$. If they are disjoint,
then the parent of $[a_{i+1},b_{i+1}]$ is the smallest ancestor of
$[a_i,b_i]$ that contains $[a_{i+1},b_{i+1}]$. If such an ancestor does
not exist, then $[a_{i+1},b_{i+1}]$ does not have a parent. Let
$A_i=\{[a_{i_1},b_{i_1}],\dots,[a_{i_m},b_{i_m}]\}$ be the set of
ancestors of $[a_i,b_i]$, where $i_1 < \dots <i_m$. It is easy to
see that $[a_{i_1},b_{i_1}] \subset \dots \subset
[a_{i_m},b_{i_m}]$. Therefore, the smallest ancestor of $[a_i,b_i]$
that contains $[a_{i+1},b_{i+1}]$ can be found by binary search,
which takes at most $O(\log k)$ time. Furthermore, assume that
$[a_{i_j},b_{i_j}]$ is the smallest ancestor, then the set of
ancestors of $[a_{i+1},b_{i+1}]$ is $A_{i+1} =
\{[a_{i_1},b_{i_1}],\dots,[a_{i_j},b_{i_j}]\}$. Based on these
observations, the algorithm can be described by the following
pseudo-code.
\begin{algorithmic}[1]
\STATE {Sort the input intervals $[a_i,b_i]$ as described above.}

\STATE {Initialize $A = \emptyset$. /* \emph{$A$ is the set of
ancestors of current interval $[a_i,b_i]$} */}

\FOR {$i=1$ to $k - 1$}
    \IF {$[a_i,b_i]$ contains $[a_{i+1},b_{i+1}]$}
        \STATE {Output $[a_i,b_i]$ is the parent of $[a_{i+1},b_{i+1}]$.}
        \STATE {Add $[a_{i+1},b_{i+1}]$ into $A$.}
    \ELSE
        \STATE{Assume that $A = \{[a_{i_1},b_{i_1}],\dots,[a_{i_m},b_{i_m}]\}$.}
        \STATE{Find the smallest interval in $A$ that contains $[a_{i+1},b_{i+1}]$.}
        \IF {the smallest interval is found}
            \STATE{Assume that the smallest interval is $[a_{i_j},b_{i_j}]$.}
            \STATE {Output $[a_{i_j},b_{i_j}]$ is the parent of $[a_{i+1},b_{i+1}]$.}
            \STATE {Set $A=\{[a_{i_1},b_{i_1}],\dots,[a_{i_j},b_{i_j}],[a_{i+1},b_{i+1}]\}$.}
        \ELSE
            \STATE {Set $A=\{[a_{i+1},b_{i+1}]\}$.}
        \ENDIF
    \ENDIF
\ENDFOR

\end{algorithmic}

\emph{Analysis:} As we argued above, the algorithm is correct. Let
us analyze the run time of the algorithm. Sorting the input
intervals takes $O(k)$ time by using integer sort since the lower
bounds are integers. It is easy to see that finding the smallest
interval from the set $A$ dominates the running time at each step of
the loop, which takes $O(\log k)$ time. Obviously, there are $k$
steps so the run time of the algorithm is $O(k\log k)$ overall.

\subsection{Algorithms for constructing the compact data
structure} \label{sec:algo-data-struct}

In this subsection, we describe two complete algorithms constructing
the data structure. The algorithms will use the routines in
subsection \ref{sec:comp-PREFIX} and subsection
\ref{sec:disjoint-intvl}. The only difference between these two algorithms
is the way of computing $PREFIX$. The first algorithm uses the
routine based on binary search to compute $PREFIX$, meanwhile, the
second one uses the trie-based routine. The following pseudo code
describes the first algorithm.
\begin{algorithmic}[1]

\FOR {$i=1$ to $n$}

    \FOR {$j=2$ to $\lambda + 1$}

        \STATE{Compute $PREFIX(s_i[j,|s_i|])$ by the routine based on binary search in Subsection \ref{sec:algo-PREFIX-BSearch}.}

    \ENDFOR

    \STATE{Output the disjoint intervals from the input intervals $PREFIX(s_i[2,|s_i|]),\dots,PREFIX(s_i[\lambda+1,|s_i|])$ by using the routine in Subsection \ref{sec:disjoint-intvl}.}

\ENDFOR

\end{algorithmic}

Let us analyze the time and memory complexity of the first
algorithm. Each computation of $PREFIX$ in line 3 takes
$O(\ell \log n)$ time and $O(1)$ extra memory. So the loop of
line 2 takes $O(\lambda \ell \log n)$ time and $O(\lambda)$ extra
memory. Computing the disjoint intervals in line 5 takes
$O(\lambda \log \lambda)$ time and $O(\lambda)$ extra memory. Since
$\lambda \leq \ell$, the run time of the loop 2 dominates the
run time of each step of loop 1. Therefore, the algorithm
takes $O(\lambda \ell n\log n)$ time and $O(\lambda)$ extra memory
in total.

The second algorithm is described by the same pseudo code above
except for the line 4 where the routine in Subsection
\ref{sec:algo-PREFIX-BSearch} computing $PREFIX(s_i[j,|s_i|])$ is
replaced by the trie-base routine in Subsection
\ref{sec:algo-PREFIX-Trie}. Let us analyze the second algorithm.
Computing $PREFIX$ in line 4 takes $O(\ell)$ time instead of $O(\ell
\log n)$ as in the first algorithm. With a similar analysis to that
of the first algorithm, the loop of line 2 takes $O(\lambda \ell n)$
time and $O(\lambda)$ extra memory. Constructing the trie in line 1
takes $O(\ell n)$ time. Therefore, the algorithm runs in $O(\lambda
\ell n)$ time. We also need $O(n)$ extra memory to store the trie.
In many cases, $n$ is much larger than $\lambda$. So the algorithm
takes $O(n)$ extra memory.

\section{Conclusions} \label{sec:conclusion}
We have described a memory efficient data structure that represents
the exact-match overlap graph. We have shown that this data
structure needs at most $(2\lambda - 1)(2\lceil \log n \rceil +
\lceil \log \lambda \rceil) n$ bits, which is a surprising result
because the number of edges in the graph can be $\Omega(n^2)$. Also,
it takes $O(\log \lambda)$ time to access an edge through the data
structure. We have proposed two fast algorithms to construct the
data structure. The first algorithm is based on binary search and
runs in $O(\lambda\ell n\log n)$ time and takes $O(\lambda)$ extra
memory. The second algorithm, based on the trie, runs in
$O(\lambda\ell n)$ time, which is slightly faster than the first
algorithm, but it takes $O(n)$ extra memory to store the trie. The
nice thing about the first algorithm is that the memory it uses is
mostly the memory of the input strings. This feature is very crucial
for building an efficient DNA assembler. Speaking of DNA assembly,
our data structure will definitely help building a DNA assembler
that can handle very large scale datasets. In the future, we would
like to exploit our data structure to speed up some operations on
the exact-match overlap graphs that are commonly used in a DNA
assembler such as removing transitive edges, greedily walking on the
graph, extracting all of the chains, etc.

\section{Acknowledgements} \label{sec:ack}
The authors would like to thank Vamsi Kundeti for discussions. The
authors also would like to thank SODA reviewers for many helpful
comments.

\bibliographystyle{alpha}
\bibliography{paper}

\end{document}